\documentclass[10pt,conference]{IEEEtran}
\usepackage{fixltx2e}
\usepackage{cite}
\usepackage{url}
\usepackage{mathrsfs}
\usepackage{color}
\usepackage{float}
\usepackage[caption=false]{subfig}

\ifCLASSINFOpdf
   \usepackage[pdftex]{graphicx}
   \graphicspath{{Figs/}}
   \DeclareGraphicsExtensions{.pdf,.jpeg,.png}
\else
\fi

\usepackage[cmex10]{amsmath}
\usepackage{amsmath}
\usepackage{amssymb}
\usepackage{amsthm}
\usepackage{amsfonts}
\usepackage{bm}
\usepackage{xfrac}
\usepackage{empheq}
\usepackage[normalem]{ulem} 
\usepackage{soul} 
\usepackage{mathtools}
\DeclarePairedDelimiter\ceil{\lceil}{\rceil}

\newtheorem{theorem}{Theorem}

\newtheorem{example}{Example}

\newtheorem{lemma}{Lemma}

\newtheorem{remark}{Remark}
\theoremstyle{definition}

\DeclareMathOperator*{\argmin}{arg\,min}
\usepackage[a4paper,bindingoffset=0.2in,%
left=1in,right=1in,top=1in,bottom=1in,%
footskip=.25in]{geometry}
\graphicspath{{figs/}}

\begin{document}
	\newgeometry{left=0.7in,right=0.7in,top=.5in,bottom=1in}
	\title{Private Variable-Length Coding with Non-zero Leakage}
\vspace{-5mm}
\author{
		\IEEEauthorblockN{Amirreza Zamani, Tobias J. Oechtering, Mikael Skoglund \vspace*{0.5em}
			\IEEEauthorblockA{\\
                             Division of Information Science and Engineering, KTH Royal Institute of Technology \\
				Email: \protect amizam@kth.se, oech@kth.se, skoglund@kth.se }}\vspace*{-2.5em}
		}
	\maketitle
%
\begin{abstract}
	A private compression design problem is studied, where an encoder observes useful data $Y$, wishes to compress it using variable length code and communicates it through an unsecured channel. Since $Y$ is correlated with private data $X$, the encoder uses a private compression mechanism to design encoded message $\cal C$ and sends it over the channel. An adversary is assumed to have access to the output of the encoder, i.e., $\cal C$, and tries to estimate $X$.
	Furthermore, it is assumed that both encoder and decoder have access to a shared secret key $W$. In this work, we generalize the perfect privacy (secrecy) assumption and consider a non-zero leakage between the private data $X$ and encoded message $\cal C$. The design goal is to encode message $\cal C$ with minimum possible average length that satisfies non-perfect privacy constraints. 
	
	We find upper and lower bounds on the average length of the encoded message using different privacy metrics and study them in special cases. For the achievability we use two-part construction coding and extended versions of Functional Representation Lemma. Lastly, in an example we show that the bounds can be asymptotically tight.

\end{abstract}
\section{Introduction}
In this paper, random variable (RV) $Y$ denotes the useful data and is correlated with the private data denoted by RV $X$. An encoder wishes to compress $Y$ and communicates it to a user over an unsecured channel. The encoded message is described by RV $\cal{C}$.  As shown in Fig.~\ref{ITWsys}, it is assumed that an adversary has access to the encoded message $\mathcal{C}$, and wants to extract information about $X$. Moreover, it is assumed that the encoder and decoder have access to a shared secret key denoted by RV $W$ with size $M$. The goal is to design encoded message $\cal C$, which compresses $Y$, using a variable length code with minimum possible average length that satisfies certain privacy constraints. We utilize techniques used in privacy mechanism and compression design problems and combine them to build such $\cal C$.
In this work, we extend previous existing results \cite{kostala,kostala2}, by generalizing the perfect privacy constraint and allowing non-zero leakage. We use different privacy leakage constraints, e.g., the mutual information between $X$ and $\cal C$ equals to $\epsilon$, strong privacy constraint and bounded per-letter privacy criterion. Here, we introduce an approach and apply it to a lossless data compression problem.    

Recently, the privacy mechanism and compression design problems are receiving increased attention
\cite{shannon, gunduz2010source, schaefer, sankar, Calmon2, Total,deniz3, yamamoto, issa, makhdoumi,borz,khodam,Khodam22,kostala, kostala2, king1, king2,asoodeh1}. 
Specifically, in \cite{shannon}, a notion of perfect secrecy is introduced by Shannon where the public data and private data are statistically independent. 
Equivocation as a measure of information leakage for information theoretic security has been used in \cite{gunduz2010source, schaefer, sankar}. 
Fundamental limits of the privacy utility trade-off measuring the leakage using estimation-theoretic guarantees are studied in \cite{Calmon2}. A privacy design with total variation as privacy measure has been studied in \cite{Total}. A related source coding problem with secrecy is studied in \cite{yamamoto}.
The concept of privacy funnel is introduced in \cite{makhdoumi}, where the privacy utility trade-off has been studied considering the log-loss as privacy measure and a distortion measure for utility. 
The privacy-utility trade-offs considering equivocation and expected distortion as measures of privacy and utility are studied in both \cite{sankar} and \cite{yamamoto}.

In \cite{borz}, the problem of privacy-utility trade-off considering mutual information both as measures of utility and privacy is studied. It is shown that under the perfect privacy assumption, the privacy mechanism design problem can be obtained by a linear program. 
\begin{figure}[]
	\centering
	\includegraphics[scale = .5]{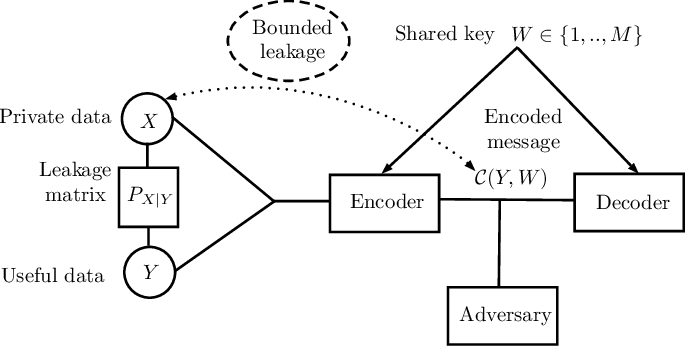}
	\caption{In this work an encoder wants to compress $Y$ which is correlated with $X$ under certain privacy leakage constraints and send it over a channel where an eavesdropper has access to the output of the encoder.} 
	\label{ITWsys}
\end{figure}
Moreover, in \cite{borz}, it has been shown that information can be only revealed if the kernel (leakage matrix) between useful data and private data is not invertible. In \cite{khodam}, the work \cite{borz} is generalized by relaxing the perfect privacy assumption allowing some small bounded leakage. More specifically, the privacy mechanisms with a per-letter privacy criterion considering an invertible kernel are designed allowing a small leakage. This result is generalized to a non-invertible leakage matrix in \cite{Khodam22}.\\
In \cite{kostala}, an approach to partial secrecy that is called \emph{secrecy by design} has been introduced and is applied to two information processing problems: privacy mechanism design and lossless compression. For the privacy design problem, bounds on privacy-utility trade-off are derived by using the Functional Representation Lemma. These results are derived under the perfect secrecy assumption.
In \cite{king1}, the privacy problems considered in \cite{kostala} are generalized by relaxing the perfect secrecy constraint and allowing some leakages. 
In \cite{king2}, the privacy-utility trade-off with two different per-letter privacy constraints is studied. 
Moreover, in \cite{kostala}, the problems of fixed length and variable length compression have been studied and upper and lower bounds on the average length of encoded message have been derived. These results are derived under the assumption that the private data is independent of the encoded message. 

Our problem here is closely related to \cite{kostala}, 
where we generalize the variable length lossless compression problem considered in \cite{kostala} by removing the assumption that $X$ and $\cal C$ are independent, i.e., $I(X;\mathcal{C}(Y,W))=0$, and therefore allowing small leakage. 

To this end we combine the privacy design techniques used in \cite{king1}, and \cite{king2}, which are based on extended versions of Functional Representation Lemma (FRL) and Strong Functional Representation Lemma (SFRL), as well as the lossless data compression design in \cite{kostala}. FRL and SFRL are constructive lemmas that are valuable for the design. We find lower and upper bounds on the average length of the encoded message $\cal C$ and study them in different scenarios. For the achievability, we use two-part construction coding. 
In an example we show that the obtained bounds can be asymptotically tight.
Furthermore, in case of perfect privacy, the existing bounds found in \cite{kostala} are improved considering the case where $X=(X_1,X_2)$ and $X_2$ is a deterministic function of $X_1$.

\section{system model and Problem Formulation} \label{sec:system}
Let $P_{XY}$ denote the joint distribution of discrete random variables $X$ and $Y$ defined on alphabets $\cal{X}$ and $\cal{Y}$. We assume that cardinality $|\mathcal{X}|$ is finite and $|\mathcal{Y}|$ is finite or countably infinite.
We represent $P_{XY}$ by a matrix defined on $\mathbb{R}^{|\mathcal{X}|\times|\mathcal{Y}|}$ and
marginal distributions of $X$ and $Y$ by vectors $P_X$ and $P_Y$ defined on $\mathbb{R}^{|\mathcal{X}|}$ and $\mathbb{R}^{|\mathcal{Y}|}$ given by the row and column sums of $P_{XY}$.  
The relation between $X$ and $Y$ is given by the leakage matrix $P_{X|Y}$ defined on $\mathbb{R}^{|\mathcal{X}|\times|\cal{Y}|}$.
The shared secret key is denoted by discrete RV $W$ defined on $\{1,..,M\}$ and is assumed to be accessible by both encoder and decoder. Furthermore, we assume that $W$ is uniformly distributed and is independent of $X$ and $Y$. 
A prefix-free code with variable length and shared secret key of size $M$ is a pair of mappings:
\begin{align*}
(\text{encoder}) \ \mathcal{C}: \ \mathcal{Y}\times \{1,..,M\}\rightarrow \{0,1\}^*\\
(\text{decoder}) \ \mathcal{D}: \ \{0,1\}^*\times \{1,..,M\}\rightarrow \mathcal{Y}.
\end{align*}
The output of the encoder $\mathcal{C}(Y,W)$ describes the encoded message. 
The variable length code $(\mathcal{C},\mathcal{D})$ is lossless if 
\begin{align}
\mathbb{P}(\mathcal{D}(\mathcal{C}(Y,W),W)=Y)=1.
\end{align}  
Similar to \cite{kostala}, we define $\epsilon$-private, strongly $\epsilon$-private and point-wise $\epsilon$-private codes.
The code $(\mathcal{C},\mathcal{D})$ is \textit{$\epsilon$-private} if
\begin{align}
I(\mathcal{C}(Y,W);X)=\epsilon.\label{lash1}
\end{align}
Moreover, the code $(\mathcal{C},\mathcal{D})$ is \textit{bounded $\epsilon$-private} if
\begin{align}
I(\mathcal{C}(Y,W);X)\leq\epsilon.\label{lash5}
\end{align}
Let $\xi$ be the support of $\mathcal{C}(W,Y)$. For any $c\in\xi$ let $\mathbb{L}(c)$ be the length of the codeword. The code $(\mathcal{C},\mathcal{D})$ is of \textit{$(\alpha,M)$-variable-length} if 
\begin{align}\label{jojo}
\mathbb{E}(\mathbb{L}(\mathcal{C}(Y,w)))\leq \alpha,\ \forall w\in\{1,..,M\}.
\end{align} 
Finally, let us define the sets $\mathcal{H}^{\epsilon}(\alpha,M)$, and $\mathcal{H}^{\epsilon}_b(\alpha,M)$, as follows: 
$\mathcal{H}^{\epsilon}(\alpha,M)\triangleq\{(\mathcal{C},\mathcal{D}): (\mathcal{C},\mathcal{D})\ \text{is}\ \epsilon\text{-private and of}\ (\alpha,M)\text{-variable-length}  \}$, and $\mathcal{H}^{\epsilon}_b(\alpha,M)\triangleq\{(\mathcal{C}',\mathcal{D}'): (\mathcal{C}',\mathcal{D}')\ \text{is bounded}\ \epsilon\text{-private and of}\ (\alpha,M)\text{-variable-length}  \}$. 
The private compression design problems can be then stated as follows
\begin{align}
\mathbb{L}(P_{XY},M,\epsilon)&=\inf_{\begin{array}{c} 
	\substack{(\mathcal{C},\mathcal{D}):(\mathcal{C},\mathcal{D})\in\mathcal{H}^{\epsilon}(\alpha,M)}
	\end{array}}\alpha,\label{main1}\\
\mathbb{L}^b(P_{XY},M,\epsilon)&=\inf_{\begin{array}{c} 
	\substack{(\mathcal{C},\mathcal{D}):(\mathcal{C},\mathcal{D})\in\mathcal{H}^{\epsilon}_b(\alpha,M)}
	\end{array}}\alpha.\label{main2}
\end{align} 
Clearly, we have $\mathbb{L}^b(P_{XY},M,\epsilon)\leq \mathbb{L}(P_{XY},M,\epsilon)$ and $\mathbb{L}^b(P_{XY},M,\epsilon)\leq \mathbb{L}(P_{XY},M,0)$.
\begin{remark}
	\normalfont 
	By letting $\epsilon=0$, both \eqref{main1} and \eqref{main2} lead to the privacy-compression rate trade-off studied in \cite{kostala}. 
	In this paper, we generalize the trade-off by considering a non-zero $\epsilon$.
\end{remark}
\begin{remark}
	The problem defined in \eqref{main1} can be studied using a different setup. To see a benefit consider the following scenario: 
	 Assume that user 1 seeks for $Y$ and we wish to compress $Y$ and send the encoded message denoted by $U$ to User 1 with a minimum possible average length. Moreover, User 2 asks for a certain amount of information about $X$. 
	For instance, let the utility of User 2 be measured by the mutual information between $X$ and the encoded message $U$, i.e., $I(U;X)\geq \Delta$. 
	In this scenario, Theorem~\ref{the1} and Theorem~\ref{the2} can be useful to find lower and upper bounds.
\end{remark}

 \section{Main Results}\label{sec:resul}
 In this section, we derive upper and lower bounds on $\mathbb{L}(P_{XY},M,\epsilon)$, and $\mathbb{L}^b(P_{XY},M,\epsilon)$, defined in \eqref{main1}, and \eqref{main2}. 
 Next, we study the bounds in different scenarios, moreover, we provide an example that shows that the bounds can be tight. 
 To do so let us recall ESFRL as follows.
 \begin{lemma}\label{EFRL}(EFRL \cite[Lemma~3]{king1}):
 	For any $0\leq\epsilon< I(X;Y)$ and pair of RVs $(X,Y)$ distributed according to $P_{XY}$, there exists a RV $U$ supported on $\mathcal{U}$ such that the leakage between $X$ and $U$ is equal to $\epsilon$, i.e., we have
 	$
 	I(U;X)= \epsilon,
 	$
 	$Y$ is a deterministic function of $(U,X)$, i.e., we have  
 	$
 	H(Y|U,X)=0,
 	$
 	and 
 	$
 	|\mathcal{U}|\leq \left[|\mathcal{X}|(|\mathcal{Y}|-1)+1\right]\left[|\mathcal{X}|+1\right].
 	$
 	Furthermore, if $X$ is a deterministic function of $Y$ we have
 	\begin{align}\label{prove}
 	|\mathcal{U}|\leq \left[|\mathcal{Y}|-|\mathcal{X}|+1\right]\left[|\mathcal{X}|+1\right].
 	\end{align}
 \end{lemma}  
\begin{proof}
	The only part which needs to be proved is \eqref{prove} and is provided in Appendix~A. 
\end{proof}
The following lemma helps us to find lower bounds on the entropy of a RV that satisfies $I(U;X)= \epsilon$ and $H(Y|U,X)=0$ considering the shared key.
\begin{lemma}\label{converse}
	Let the pair $(X,Y)$ be jointly distributed and RV $W$ be independent of $X$ and $Y$. Then, if RV $U$ satisfies 
	$
		I(U;X)= \epsilon,
	$
	and 
	$
	H(Y|U,X,W)=0,
	$
	then, we have
	\begin{align}
	H(U)\!\geq \!\max\{L_1(\epsilon,P_{XY}),L_2(\epsilon,P_{XY}), L_3(\epsilon,P_{XY})\},
	\end{align}
	where
	\begin{align}
	L_1(\epsilon,P_{XY})&=H(Y|X),\label{khar}\\
	L_2(\epsilon,P_{XY})&=\min_{x\in\mathcal{X}} H(Y|X)+\epsilon,\label{khar1}\\
	L_3(\epsilon,P_{XY})&=H(Y|X)-H(X|Y)+\epsilon\label{khar2},
	\end{align}
	and $\min_{x\in\mathcal{X}} H(Y|X=x)$ is the minimum conditional entropy which is non-zero.
\end{lemma}
\begin{proof}
	The proof is provided in Appendix~A.
\end{proof}
In the next two theorems we provide upper and lower bounds on $\mathbb{L}(P_{XY},M,\epsilon)$. The next theorem is a generalization of \cite[Theorem~8]{kostala} for correlated $\cal C$ and $X$. 
\begin{theorem}\label{the1}
	Let $0\leq \epsilon \leq H(X)$ and the pair of RVs $(X,Y)$ be distributed according to $P_{XY}$, the shared secret key size be $|\mathcal{X}|$, i.e., $M=|\mathcal{X}|$. Then, we have
	\begin{align}
	\mathbb{L}(P_{XY},|\mathcal{X}|,\epsilon)&\leq\nonumber\\ & \!\!\!\sum_{x\in\mathcal{X}}\!\!H(Y|X=x)\!+\!\epsilon\!+\!h(\alpha)\!+\!1+\!\ceil{\log (|\mathcal{X}|)},\label{koon}
	\end{align}
	where $\alpha=\frac{\epsilon}{H(X)}$ and if $|\mathcal{Y}|$ is finite we have
	\begin{align}
	&\mathbb{L}(P_{XY},|\mathcal{X}|,\epsilon)\nonumber\\& \leq \ceil{\log\left(|\mathcal{X}|(|\mathcal{Y}|-1)+1\right)\left(|\mathcal{X}|+1\right)}+\ceil{\log (|\mathcal{X}|)},\label{koon2}
	\end{align}
	Finally, if $X$ is a deterministic function of $Y$ we have
	\begin{align}\label{goh}
	&\mathbb{L}(P_{XY},|\mathcal{X}|,\epsilon)\nonumber\\& \leq \ceil{\log(\left(|\mathcal{Y}|-|\mathcal{X}|+1\right)\left[|\mathcal{X}|+1\right])}+\ceil{\log (|\mathcal{X}|)}.
	\end{align}
\end{theorem}
\begin{proof}
	The proof is based on \cite[Lemma~5]{king1} and the two-part construction coding and is provided in Appendix~A.
\end{proof}
In the next theorem we provide lower bounds on $\mathbb{L}(P_{XY},M,\epsilon)$.
\begin{theorem}\label{the2}
	Let $0\leq \epsilon \leq H(X)$ and the pair of RVs $(X,Y)$ be distributed according to $P_{XY}$ supported on alphabets $\mathcal{X}$ and $\mathcal{Y}$, where $|\mathcal{X}|$ is finite and $|\mathcal{Y}|$ is finite or countably infinite. For any shared secret key size $M\geq 1$ we have
	\begin{align}\label{23}
	&\mathbb{L}(P_{XY},M,\epsilon)\geq\nonumber\\ &\max\{L_1(\epsilon,P_{XY}),L_2(\epsilon,P_{XY}),L_3(\epsilon,P_{XY})\},
	\end{align}
	where $L_1(\epsilon,P_{XY})$, $L_2(\epsilon,P_{XY})$ and $L_3(\epsilon,P_{XY})$ are defined in \eqref{khar}, \eqref{khar1}, and \eqref{khar2}. 
	Furthermore, if $X$ is deterministic function of $Y$, then we have
	\begin{align}\label{ma}
	\mathbb{L}(P_{XY},M,\epsilon)\geq \log(\frac{1}{\max_x P_X(x)}).
	\end{align} 
\end{theorem}
\begin{proof}
	The proof is based on Lemma~\ref{converse} and is provided in Appendix~A.
\end{proof}
Next, we provide an example where the upper and lower bounds obtained in Theorem~\ref{the1} and Theorem~\ref{the2} are studied. 
\begin{example}\label{anus}
	Let $Y=(Y_1,..,Y_N)$ be an i.i.d Bernoulli($0.5$) sequence and $X_N=\frac{1}{N}\sum_{i=1}^{N} Y_i$. Clearly, $X_N$ is a function of $Y$, $|\mathcal{X}|=N+1$ and $|\mathcal{Y}|=2^N$. Using \cite[(72)]{kostala} we have
	\begin{align}
	h_0(P_{XY})\stackrel{(a)}{=}H(Y|X)=\sum_{i=1}^{N} (\frac{1}{2})^N \binom{N}{i} ,
	\end{align}
	where (a) follows from \cite[(70)]{kostala}. 
	Let the shared key size be $N+1$, by using Theorem~\ref{the1} and Theorem~\ref{the2} we have
	\begin{align}
	\epsilon+\sum_{i=1}^{N} (\frac{1}{2})^N \binom{N}{i}&\leq\mathbb{L}(P_{XY},N+1,\epsilon)\nonumber\\&\leq \log((N\!+\!2)(2^N\!-\!N))\!+\!\log(N\!+\!1),\label{kir}
	\end{align}
	where \eqref{goh} is used for the upper bound and $L_3(\epsilon,P_{XY})$ is used for the lower bound. Using \cite[(73)]{kostala} we have
	\begin{align}\label{jee}
	\lim_{N\rightarrow \infty} \frac{h_0(P_{XY})}{N}= \lim_{N\rightarrow \infty} \frac{H(Y|X)}{N}=h(0.5)=1,
	\end{align}
	where $h(\cdot)$ is the binary entropy function. 
	Now by using \eqref{kir} and \eqref{jee} we obtain
	\begin{align}
	\lim_{N\rightarrow \infty} \frac{\mathbb{L}(P_{XY},N+1,\epsilon)}{N}=1
	\end{align}
\end{example} 
 \textbf{\emph{Bounds for bounded leakage constraint: }}
 Here, we assume that $X=(X_1,X_2)$ where $\mathcal{X}=\mathcal{X}_1\times\mathcal{X}_2$, $1<|\mathcal{X}_1|\leq |\mathcal{X}_2|$, and $|\mathcal{X}_1||\mathcal{X}_2|=|\mathcal{X}|$. In the next theorem we provide upper bounds for $\mathbb{L}^b(P_{XY},M,\epsilon)$. We show that when the leakage is greater than a threshold we are able to communicate the message over the channel using a shared key size less than $|\mathcal{X}|$ and the receiver can decode the message without any loss. 
 \begin{theorem}\label{the7}
 	Let $\epsilon\geq H(X_1)$ and RVs $(X_1,X_2,Y)$ be distributed according to $P_{X_1X_2Y}$ and let the shared secret key size be $|\mathcal{X}_2|$, i.e., $M=|\mathcal{X}_2|$. Then, we have
 	\begin{align}
 	&\mathbb{L}^b(P_{XY},|\mathcal{X}_2|,\epsilon)\leq\nonumber\\ & \!\!\!\sum_{x\in\mathcal{X}}\!\!H(Y|X=x)\!+H(X_1)\!+\!2+\!\ceil{\log (|\mathcal{X}_2|)},\label{koskoon}
 	\end{align}
 	if $|\mathcal{Y}|$ is finite, then 
 	\begin{align}
 	&\mathbb{L}^b(P_{XY},|\mathcal{X}_2|,\epsilon)\nonumber\\& \leq \ceil{\log\left(|\mathcal{X}|(|\mathcal{Y}|-1)+1\right)}+\ceil{\log (|\mathcal{X}_2|)}+2+H(X_1),\label{koskoon3}
 	\end{align}
 	and if $X=(X_1,X_2)$ is deterministic function of $Y$, we get
 	\begin{align}
 	&\mathbb{L}^b(P_{XY},|\mathcal{X}_2|,\epsilon)\nonumber\\& \leq \ceil{\log\left(|\mathcal{Y}|-|\mathcal{X}|+1\right)}+\ceil{\log (|\mathcal{X}_2|)}+2+H(X_1),\label{koskoon4}
 	\end{align}
 	Furthermore, let $\epsilon\geq H(X_2)$ and the shared secret key size be $|\mathcal{X}_1|$. We have
 	\begin{align}
 	&\mathbb{L}^b(P_{XY},|\mathcal{X}_1|,\epsilon)\leq\nonumber\\ & \!\!\!\sum_{x\in\mathcal{X}}\!\!H(Y|X=x)\!+H(X_2)\!+\!2+\!\ceil{\log (|\mathcal{X}_1|)},\label{koskoon2}
 	\end{align}
 \end{theorem}
 \begin{proof}
 	The proof is similar as proof of Theorem~\ref{the1}. For proving \eqref{koskoon}, $X_2$ is encoded by one-time-pad coding \cite[Lemma~1]{kostala2}, which uses $\ceil{\log(|\mathcal{X}_2|)}$ bits. Let the output of one-time-pad coding be $\tilde{X}_2$. Next we encode $X_1$ using any traditional code which uses at most $H(X_1)+1$ bits. Next, we produce $U$ based on FRL \cite[Lemma~1]{kostala} with assignments $X\leftarrow (X_1,X_2)$ and $Y\leftarrow Y$, and encode it using any traditional code which uses at most $H(U)+1$ bits. We send encoded $X_1$, $\tilde{X}_2$ and $U$ over the channel. Since $U$ satisfies $I(U;X_1,X_2)=0$ and $H(Y|U,X_1,X_2)=0$ by using the upper bound on $H(U)$ found in \cite[Lemma~2]{kostala} we obtain 
 	$
 	H(U)\leq \sum_x H(Y|X=x).
 	$
 	Finally, we obtain
 	$
 	\mathbb{L}^b(P_{XY},|\mathcal{X}_2|,\epsilon)\leq \sum_{x\in\mathcal{X}}\!\!H(Y|X=x)\!+H(X_1)\!+\!2+\!\ceil{\log (|\mathcal{X}_2|)}.
 	$ 
 	Note that at receiver side, using the shared key and $\tilde{X}_2$ we can decode $X_2$ and by using $X_1$, $X_2$ and $U$ we can decode $Y$ without loss. For the leakage constraint we note that the one-time-pad coding can be used such that $\tilde{X}_2$ be independent of $(X_1,X_2,U)$. Thus, we have
 	$
 	I(X_1,\tilde{X}_2,U;X_1,X_2)=H(X_1)+H(\tilde{X}_2)+H(U)-H(X_1|X_1,X_2)-H(\tilde{X}_2|X_1,X_2)-H(U|X_1,X_2,\tilde{X}_2)\stackrel{(a)}{=}H(X_1)\leq \epsilon,
 	$   
 	where (a) follows by the independency between $\tilde{X}_2$ and $(U,X_1,X_2)$, and independency between $U$ and $(X_1,X_2)$. 
 	Proof of \eqref{koskoon3} and \eqref{koskoon4} are based on \cite[(31)]{kostala} and \cite[(41)]{kostala}.
 	Proof of \eqref{koskoon2} is similar to \eqref{koskoon}. We use one-time-pad coding for encoding $X_1$ instead of $X_2$. In this case we send compressed $X_2$, $\tilde{X}_1$ and $U$ over the channel, where $\tilde{X}_1$ is the output of one-time-pad coding. We then have
 	$
 	\mathbb{L}^b(P_{XY},|\mathcal{X}_1|,\epsilon)\leq \sum_{x\in\mathcal{X}}\!H(Y|X=x)\!+H(X_2)\!+\!2+\!\ceil{\log (|\mathcal{X}_1|)}.
 	$
 \end{proof}
 \begin{remark}
 	As argued in \cite[Theorem~9]{kostala}, when the leakage is zero and $X$ is a deterministic function of $Y$, if the shared key size is less than the size of $X$, i.e., $M\leq |\mathcal{X}|$, lossless variable length codes do not exist. However, as it is shown in Theorem~\ref{the7}, when the leakage is more than a threshold, for the key size less than $|\mathcal{X}|$ such codes exist.
 \end{remark}
 \begin{remark}\label{tt}
 	The upper bounds \eqref{koskoon} and \eqref{koskoon2} can be less than the upper bound found in \cite[Theorem~9]{kostala} which is obtained under perfect privacy assumption. For instance, if $H(X_1)+2\leq \log(|\mathcal{X}_1|)$, then \eqref{koskoon} is less than \cite[(95)]{kostala}. Later follows since we have
 	$
 	H(X_1)+1+\ceil{\log (|\mathcal{X}_2|)}\leq \log(|\mathcal{X}_1|)+\ceil{\log (|\mathcal{X}|)-\log (|\mathcal{X}_1|)}-1 
 	\leq\ceil{\log(|\mathcal{X}_1|)}+\ceil{\log (|\mathcal{X}|)-\log (|\mathcal{X}_1|)}-1 \leq \ceil{\log (|\mathcal{X}|)},
 	$  
 	where the last inequality follows by $\ceil{a}+\ceil{b}\leq \ceil{a+b}+1$.
 \end{remark}
 In the next theorem we find lower bounds for $\mathbb{L}^b(P_{XY},M,\epsilon)$.
 \begin{theorem}\label{the8}
 	Let $0\leq \epsilon \leq H(X)$ and RVs $(X_1,X_2,Y)$ be distributed according to $P_{X_1X_2Y}$. For any shared secret key size $M\geq 1$ we have
 	\begin{align}\label{23456}
 	\mathbb{L}^b(P_{XY},M,\epsilon)\geq H(Y|X)-H(X|Y).
 	\end{align}
 	Furthermore, if $X$ is deterministic function of $Y$, then we have
 	\begin{align}\label{ma26}
 	\mathbb{L}^b(P_{XY},M,\epsilon)\geq \log(\frac{1}{\max_x P_X(x)}).
 	\end{align}
 \end{theorem}
 \begin{proof}
 	Let $U=\mathcal{C}(Y,W)$ be the output of the encoder. The code is bounded $\epsilon$-private, thus $I(U;X_1,X_2)\leq \epsilon$. Due to recoverability, i.e., given the code, shared secret key, and the private data, we can recover the useful data, it also satisfies $H(Y|W,U,X_1,X_2)=0$. Using \eqref{koskesh} we have
 	\begin{align*}
 	H(U)&\geq I(U;Y,W)\\&=I(U;X)+H(Y|X)+H(W|X,Y)\\&-H(W|X,U)-H(Y|X,U,W)-H(X|Y,W)\\&+H(X|W,Y,U)\\&\geq H(Y|X)+H(W)-H(W|X,U)\\&-H(X|Y)+H(X|Y,U,W)\\&\geq H(Y|X)-H(X|Y),
 	\end{align*}
 	where we used $I(U;X)\geq 0$, $H(X|Y,U,W)\geq 0$ and independency of $W$ and $(X,Y)$. Proof of \eqref{ma26} is similar to \eqref{ma}. 
 \end{proof}
 \textbf{\emph{A general approach for $X$:}}
 Next, by using the same construction as in Theorem~\ref{the7}, we find upper bounds on $\mathbb{L}^b(P_{XY},M,\epsilon)$ for any joint distribution $P_{XY}$. 
 Note that if $|\mathcal{X}|$ is not a prime number then we can show $X$ by $(X_1,X_2)$ where $\mathcal{X}=\mathcal{X}_1\times\mathcal{X}_2$, $1<|\mathcal{X}_1|\leq |\mathcal{X}_2|$, $|\mathcal{X}_1||\mathcal{X}_2|=|\mathcal{X}|$ and $P_X(x)=P_{X_1,X_2}(x_1,x_2)$. Furthermore, if $|\mathcal{X}|$ is a prime number then we can show $X$ by $(X_1,X_2)$ where $|\mathcal{X}_1||\mathcal{X}_2|=|\mathcal{X}|+1$ and $P_{X_1,X_2}(|\mathcal{X}_1|,|\mathcal{X}_2|)=0$. Let $\mathcal{S}_X$ be all possible representations of $X$ where $X=(X_1,X_2)$.  
 For a fixed $\epsilon$ we define $\mathcal{S}_X^{1\epsilon}=\{(X_1,X_2): (X_1,X_2)\in \mathcal{S}_X, H(X_1)\leq \epsilon\}$ and $\mathcal{S}_X^{2\epsilon}=\{(X_1,X_2): (X_1,X_2)\in \mathcal{S}_X, H(X_2)\leq \epsilon\}$. 
 \begin{theorem}\label{the9}
 	For any $\epsilon\geq 0$, pair of RVs $(X,Y)$ distributed according to $P_{XY}$ and shared key size $M\geq \alpha$, if $\mathcal{S}_X^{1\epsilon}\neq \O $ we have
 	\begin{align}
 	&\mathbb{L}^b(P_{XY},M,\epsilon)\leq\nonumber\\ & \!\!\!\sum_{x\in\mathcal{X}}\!\!H(Y|X=x)\!+\!2\!+\!\min_{(X_1,X_2)\in\mathcal{S}_X^{1\epsilon}}\left\{H(X_1)\!+\!\ceil{\log (|\mathcal{X}_2|)}\right\},\label{koskoon6}
 	\end{align}
 	where $\alpha = |\argmin_{X_2:(X_1,X_2)\in \mathcal{S}_X^{1\epsilon}} H(X_1)+\ceil{\log (|\mathcal{X}_2|)}|<|\mathcal{X}|$ and $|S|$ denotes the size of the RV $S$. If $\mathcal{S}_X^{2\epsilon}\neq \O $, for a shared key size $M\geq \beta$ we have
 	\begin{align}
 	&\mathbb{L}^b(P_{XY},M,\epsilon)\leq\nonumber\\ & \!\!\!\sum_{x\in\mathcal{X}}\!\!H(Y|X=x)\!+\!2\!+\!\min_{(X_1,X_2)\in\mathcal{S}_X^{1\epsilon}}\left\{H(X_2)\!+\!\ceil{\log (|\mathcal{X}_1|)}\right\},\label{koskoon7}
 	\end{align}
 	where $\alpha = |\argmin_{X_1:(X_1,X_2)\in \mathcal{S}_X^{2\epsilon}} H(X_1)+\ceil{\log (|\mathcal{X}_2|)}|<|\mathcal{X}|$.
 \end{theorem}
 \begin{proof}
 	For proving \eqref{koskoon6} let $(X_1,X_2)\in\mathcal{S}_X^{1\epsilon}$. Then, by using Theorem~\ref{the7} we have
 	\begin{align}
 	&\mathbb{L}^b(P_{XY},|\mathcal{X}_2|,\epsilon)\leq\nonumber\\ & \!\!\!\sum_{x\in\mathcal{X}}\!\!H(Y|X=x)\!+H(X_1)\!+\!2+\!\ceil{\log (|\mathcal{X}_2|)}.\label{an}
 	\end{align}
 	Since \eqref{an} holds for any $(X_1,X_2)\in\mathcal{S}_X^{1\epsilon}$ we can take minimum over $(X_1,X_2)$ which results in \eqref{koskoon6}. 
 	Furthermore, the key size is chosen as $|\argmin_{X_2:(X_1,X_2)\in \mathcal{S}_X^{1\epsilon}} H(X_1)+\ceil{\log (|\mathcal{X}_2|)}|<|\mathcal{X}|$ and note that the first term $\sum_{x\in\mathcal{X}}H(Y|X=x)$ remains the same since it is a function of $P_{Y|X_1X_2}$, i.e., distributions including the pair $(X_1,X_2)$ do not change since $P_{X_1,X_2}(x_1,x_2)=P_X(x)$.
 	Using \eqref{koskoon2} and following the same approach \eqref{koskoon7} is obtained.
 \end{proof}
 \begin{remark}
 	The upper bounds \eqref{koskoon6} and \eqref{koskoon7} generalize \eqref{koskoon} and \eqref{koskoon2}. Upper bounds \eqref{koskoon} and \eqref{koskoon2} are obtained for fixed $(X_1,X_2)$. Thus, by taking minimum over $(X_1,X_2)\in\mathcal{S}_X^{1\epsilon}$ and $(X_1,X_2)\in\mathcal{S}_X^{2\epsilon}$ we get tighter upper bounds.
 \end{remark}
 \begin{remark}
 	Same lower bounds as obtained in Theorem~\ref{the8} can be used here since the separation technique does not improve the bounds in \eqref{23456} and \eqref{ma26}.  
 \end{remark}
 \begin{remark}
 	Using Example~\ref{anus}, we can check that the bounds obtained in Theorem~\ref{the7}, Theorem~\ref{the8}, Theorem~\ref{the9}, and Theorem~\ref{the10} can be asymptotically tight. 
 \end{remark}
 Next, we provide an example to clarify the construction in Theorem~\ref{the9} and compare the bounds with perfect privacy case.
 \begin{example}
 	Let RV $X\in\{1,..,12\}$ with distribution $P_X(x)=0.05,\ x\in\{1,..,10\}$ and $P_X(x)=0.475,\ x\in\{11,12\}$ be correlated with $Y$ where $Y|X=x$ is a $\text{BSC}(0.5)$. To calculate the bound in Theorem~\ref{the9} one possible separation of $X$ is $(X_1,X_2)$ where distribution of $X_1$ and $X_2$ are $P_{X_1}(x_1)=0.01,\ x\in\{1,..,5\}$, $P_{X_1}(x_1)=0.95,\ x_1=6$ and $P_{X_2}(x_2)=0.5, \ x\in\{1,2\}$. In this case, $H(X_1)=0.4025$ and $\log(|\mathcal{X}_2|)=1$. 
 	Using the bound \eqref{koskoon}, for all $\epsilon\geq 0.4025$ and key size $M\geq 1$ we have 
 	\begin{align*}
 	&\mathbb{L}^b(P_{XY},M,\epsilon)\leq\nonumber\\ & \!\!\!\sum_{x\in\mathcal{X}}\!\!H(Y|X=x)\!+H(X_1)\!+\!2+\!\ceil{\log (|\mathcal{X}_2|)}=15.45\ \text{bits},
 	\end{align*}
 	By using the bound \cite[(95)]{kostala}, for the key size $M\geq 12$ we have
 	\begin{align*}
 	\mathbb{L}^b(P_{XY},M,0)\leq 17\ \text{bits}. 
 	\end{align*}
 	Note that since $|\mathcal{Y}|=2$ we can also use the following bound instead of \eqref{koskoon6}
 	\begin{align*}
 	&\mathbb{L}^b(P_{XY},M,\epsilon)\leq
 	\ceil{\log\left(|\mathcal{X}|(|\mathcal{Y}|-1)+1\right)}+2\\&+\min_{(X_1,X_2)\in\mathcal{S}_X^{1\epsilon}}\left\{\ceil{\log (|\mathcal{X}_2|)}+H(X_1)\right\}.
 	\end{align*}
 	By checking all possible separations, the minimum of $\ceil{\log (|\mathcal{X}_2|)}+H(X_1)$ occurs at $|\mathcal{X}_1|=6$, and $|\mathcal{X}_2|=2$ where $X_1$ and $X_2$ have the distributions mentioned before.
 	Thus, for all $\epsilon\geq 0.4025$ and key size $M\geq 1$ we have
 	$
 	\mathbb{L}^b(P_{XY},M,\epsilon)\leq 7.4025\ \text{bits}. 
 	$
 	By using the bound \cite[(96)]{kostala}, for the key size $M\geq 12$ we have
 	$
 	\mathbb{L}^b(P_{XY},M,0)\leq 9\ \text{bits}. 
 	$	 
 	We can see that by allowing small amount of leakage with shorter secret key size we are able to send shorter codewords.
 \end{example}
 \textbf{\emph{Perfect privacy ($\epsilon=0$):}}
 In this part, we let $\epsilon=0$ and consider the problem as in \cite{kostala}. We show that when $X=(X_1,X_2)$ where $X_2$ is a deterministic function of $X_1$ the upper bounds on $\mathbb{L}^b(P_{XY},M,0)$ can be improved by using shorter shared key size, i.e., $M<|\mathcal{X}|$. In the next theorem we provide upper bounds on $\mathbb{L}^b(P_{XY},M,0)$.
 \begin{theorem}\label{the10}
 	Let $X=(X_1,X_2)$ where $X_2=f(X_1)$ and the shared key size be $|\mathcal{X}_1|$. We have
 	\begin{align}
 	&\mathbb{L}^b(P_{XY},|\mathcal{X}_1|,0)\leq\nonumber\\ 
 	&\!\!\!\sum_{x_1\in\mathcal{X}_1}\!\!H(Y|X_1=x_1)\!+\!1\!+\!\ceil{\log (|\mathcal{X}_1|)},\label{koskoon10}
 	\end{align}
 	if $|\mathcal{Y}|$ is finite, then 
 	\begin{align}
 	&\mathbb{L}^b(P_{XY},|\mathcal{X}_1|,0)\nonumber\\& \leq \ceil{\log\left(|\mathcal{X}|(|\mathcal{Y}|-1)+1\right)}+\ceil{\log (|\mathcal{X}_1|)}+1,\label{koskoon11}
 	\end{align}
 	and if $X=(X_1,X_2)$ is deterministic function of $Y$, we get
 	\begin{align}
 	\mathbb{L}^b(P_{XY},|\mathcal{X}_1|,0) \leq \ceil{\log\left(|\mathcal{Y}|-|\mathcal{X}|+1\right)}\!+\!\ceil{\log (|\mathcal{X}_1|)}\!+\!1.\label{koskoon12}
 	\end{align}
 \end{theorem}  
 \begin{proof}
 	For proving \eqref{koskoon10}, let $U$ be produced based on FRL \cite[Lemma~1]{kostala} with assignments $X\leftarrow (X_1,X_2)$ and $Y\leftarrow Y$, and encode it using any traditional code which uses at most $H(U)+1$ bits. Thus, $I(U;X_1,X_2)=0$ and $H(Y|X_1,X_2,U)=0$. Since $X_2=f(X_1)$, we have 
 	\begin{align}\label{ann}
 	H(Y|X_1,X_2,U)=H(X_1,U)=0.
 	\end{align}
 	Using \cite[Lemma~2]{kostala}, we have 
 	$
 	H(U)\leq \sum_x H(Y|X=x).
 	$
 	We encode $X_1$ by one-time-pad coding which uses $\ceil{\log(|\mathcal{X}_1|)}$ bits. Let the output of one-time-pad coding be $\tilde{X}_1$. We send compressed $\tilde{X}_1$ and $U$ over the channel. 
 	Using the upper bound on $H(U)$ and $\tilde{X}_1$ we have
 	$
 	\mathbb{L}^b(P_{XY},|\mathcal{X}_1|,0)\leq\sum_{x\in\mathcal{X}}\!\!H(Y|X=x)+1+\ceil{\log (|\mathcal{X}_1|)}.
 	$
 	At receiver side, using the shared key, first $X_1$ is decoded and we then use \eqref{ann} to decode $Y$ using $X_1$ and $U$. For the leakage constraint, independency between $\tilde{X}_1$ and $X_1$ implies $I(X_1,X_2;\tilde{X}_1)=0$, and the later follows since $X_2=f(X_1)$. Furthermore, we choose $\tilde{X}_1$ to be independent of $U$. Thus, 
 	\begin{align*}
 	I(\tilde{X}_1,U;X_1,X_2)=I(\tilde{X}_1,U;X_1)=I(\tilde{X}_1;X_1)=0.
 	\end{align*}
 	The bounds \eqref{koskoon11} and \eqref{koskoon12} can be obtained using \cite[(31)]{kostala} and \cite[(41)]{kostala}.  
 \end{proof}
 \begin{remark}
 	Clearly, when $X=(X_1,X_2)$ with $X_2=f(X_1)$ the upper bounds \eqref{koskoon10}, \eqref{koskoon11} and \eqref{koskoon12} improve the bounds in \cite[Theorem~8]{kostala}.
 \end{remark}
 \textbf{\emph{Special case: $X=(X_1,X_2)$ and $X_2=f(X_1)$: }}
 In this part, we use similar approaches as used in Theorem~\ref{the9} and Theorem~\ref{the10} to find upper bounds. We encode $X_1$ using one-time-pad coding and same $U$ as in Theorem~~\ref{the10}. As shown in Theorem~\ref{the10}, $X_1$ is sufficient to decode $Y$ since $X_2=f(X_1)$, however in this scheme we do not use the possibility that we are allowed to leak about $X$. Similar to Theorem~\ref{the10}, we obtain
 $
 \mathbb{L}^b(P_{XY},|\mathcal{X}_1|,\epsilon)\leq\sum_{x\in\mathcal{X}}\!\!H(Y|X=x)\!+\!1\!+\!\ceil{\log (|\mathcal{X}_1|)}=
 \sum_{x_1\in\mathcal{X}_1}\!\!H(Y|X_1=x_1)\!+\!1\!+\!\ceil{\log (|\mathcal{X}_1|)}.
 $
 Next, we use the separation technique as used in Theorem~\ref{the9}. For any $\epsilon\geq 0$ if a separation exist such that $\epsilon\geq H(X_1')$, where $X_1=(X_1',X_1'')$, we obtain
 $
 \mathbb{L}^b(P_{XY},|\mathcal{X}''|,\epsilon)\leq\sum_{x_1\in\mathcal{X}_1}\!\!H(Y|X_1=x_1)+H(X_1')+2+\ceil{\log (|\mathcal{X}_1''|)}.
 $
 Furthermore, we can take minimum over all possible separations to improve the upper bound. Let $\mathcal{S}_{X_1}^\epsilon$ be the all possible separations where $\epsilon\geq H(X_1')$. We have
 \begin{align*}
 &\mathbb{L}^b(P_{XY},\alpha,\epsilon)\leq\nonumber\\ & \!\!\!\sum_{x_1\in\mathcal{X}_1}\!\!\!\!\!H(Y|X_1=x_1\!)\!+\!2\!+\!\!\!\!\!\!\min_{(X_1',X_1'')\in\mathcal{S}_{X_1}^{\epsilon}}\!\!\!\!\!\!\left\{\!H(X_1')\!+\!\ceil{\log (|\mathcal{X}_1''|)}\!\right\},
 \end{align*}
 where $\alpha = |\argmin_{X_1'':(X_1',X_1'')\in \mathcal{S}_{X_1}^{\epsilon}} H(X_1')+\ceil{\log (|\mathcal{X}_1''|)}|<|\mathcal{X}_1|$. \\
 We can see that the upper bound on $\mathbb{L}^b(P_{XY},|\mathcal{X}''|,\epsilon)$ can be less than the bound in \eqref{koskoon10}. For instance, let $H(X_1')+2\leq \log(|\mathcal{X}'_1|)$. Using the same arguments as in Remark~\ref{tt} we have $H(X_1')+1+\ceil{\log (|\mathcal{X}_1''|)}\leq \ceil{\log (|\mathcal{X}_1|)}$. 
 Hence, if the leakage is more than a threshold it can help us to improve the upper bound and send shorter codewords compared to the perfect privacy case. Moreover, it helps us to use shorter shared secret key size. 
\bibliographystyle{IEEEtran}
\bibliography{IEEEabrv,IZS}
\section*{Appendix A}
\textbf{\emph{Proof of Lemma~\ref{EFRL}:}}
The proof of \eqref{prove} is based on the construction used in \cite[Lemma~3]{king1} and \cite[Lemma~1]{kostala}. When $X$ is a deterministic function of $Y$, $\tilde{U}$ is found by FRL \cite[Lemma~1]{kostala} as it is used in 
\cite[Lemma~3]{king1} satisfies $|\mathcal{\tilde{U}}|\leq |\mathcal{Y}|-|\mathcal{X}|+1$. Using the construction as in \cite[Lemma~3]{king1} we have $|\mathcal{U}|\leq (|\mathcal{Y}|-|\mathcal{X}|+1)(|\mathcal{X}|+1)$.

\textbf{\emph{Proof of Lemma~\ref{converse}:}}
First we obtain $L_3(\epsilon,P_{XY})$. Using the key equation \cite[(7)]{king1} with assignments $X\leftarrow X$ and $Y\leftarrow(Y,W)$, for any RVs $X$, $Y$, $W$, and $U$, we have
\begin{align}\label{koskesh}
I(U;Y,W)&=I(U;X)+H(Y,W|X)\nonumber\\&-H(Y,W|X,U)-I(X;U|Y,W).
\end{align}
Using \eqref{koskesh} if $U$ satisfies $I(U;X)= \epsilon$ and $H(Y|U,X,W)=0$ and $W$ is independent of $(X,Y)$, we have
\begin{align*}
H(U)&\geq I(U;Y,W)\\&=I(U;X)+H(Y|X)+H(W|X,Y)\\&-H(W|X,U)-H(Y|X,U,W)-H(X|Y,W)\\&+H(X|W,Y,U)\\&\stackrel{(a)}{=}\epsilon+H(Y|X)+H(W)-H(W|X,U)\\&-H(X|Y)+H(X|Y,U,W)\\&\stackrel{(b)}{\geq} \epsilon+H(Y|X)-H(X|Y),
\end{align*} 
where in (a) we used $I(X;U)=\epsilon$, independence of $W$ and $(X,Y)$, and $H(Y|X,U,W)=0$. Step (b) follows from the non-negativity of $H(W)-H(W|X,U)=I(W;X,U)$ and $H(X|Y,U,W)$. 

For proving 
$L_2(\epsilon,P_{XY})$ note that we have
\begin{align}\label{tokhm}
H(U)\geq \min_x H(U|X=x)+\epsilon,
\end{align}
since 
$
H(U)\!-\!\min_x H(U|X\!\!=\!x)\!\geq\! H(U)\!-\!\!\sum_x\! P_X(x)H(U|X\!\!=\!x)
=\epsilon.
$
Furthermore,
\begin{align}
H(U|X=x)&\geq H(U|W,X=x)\label{jama}\\ &\stackrel{(a)}{\geq} H(U|W,X=x)+H(Y|U,W,X=x)\nonumber\\&= H(U,Y|W,X=x)\nonumber\\&=H(Y|W,X=x)+H(U|Y,W,X=x)\nonumber\\&\stackrel{(b)}{\geq} H(Y|W,X=x)=H(Y|X=x), \label{jata}
\end{align}
where (a) used the fact that $H(Y|U,W,X=x)=0$ since $H(Y|U,W,X)=0$. (b) follows since $W$ is independent of $X$ and $Y$. Taking the minimum at both sides we get $\min_x H(U|X=x)\geq \min_x H(Y|X=x)$. Combining it with \eqref{tokhm} we obtain $L_2(\epsilon,P_{XY})$. 
To obtain $L_1(\epsilon,P_{XY})$, we take the average from both sides of \eqref{jama} and \eqref{jata}. This gives us
$
H(U)\geq H(U|X)\geq H(Y|X),
$
which completes the proof.\\
\textbf{\emph{Proof of Theorem~\ref{the1}:}}
For proving \eqref{koon}, we use two-part construction coding scheme as used in \cite{kostala}. First, the private data $X$ is encoded by one-time-pad coding \cite[Lemma~1]{kostala2}, which uses $\ceil{\log(|\mathcal{X}|)}$ bits. Next, we produce $U$ based on EFRL \cite[Lemma~3]{king1} and encode it using any traditional code which uses at most $H(U)+1$ bits. By using the upper bound on $H(U)$ found in \cite[Lemma~5]{king1} we obtain 
$
\mathbb{L}(P_{XY},|\mathcal{X}|,\epsilon)\leq \sum_{x\in\mathcal{X}}\!\!H(Y|X=x)\!+\!\epsilon\!+\!h(\alpha)\!+\!1+\!\ceil{\log (|\mathcal{X}|)}.
$
For proving \eqref{koon2} and \eqref{goh} we use the same coding and note that when $|\mathcal{Y}|$ is finite we use the bound in \eqref{prove}, which results in \eqref{koon2}. Furthermore, when $X$ is a deterministic function of $Y$ we can use the bound in \eqref{prove} that leads to \eqref{goh}. Moreover, for the leakage constraint we note that the randomness of one-time-pad coding is independent of $X$ and the output of the EFRL. Let $Z$ be the compressed output of the EFRL and $\tilde{X}$ be the output of the one-time-pad coding. Then, we have
$
I(X;\tilde{X},Z)=I(X;\tilde{X})+I(X;Z|\tilde{X})=I(X;Z)=\epsilon,
$
where we used the independency between $\tilde{X}$ and $(Z,X)$.\\
\textbf{\emph{Proof of Theorem~\ref{the2}:}}
Let $\tilde{U}=\mathcal{C}(Y,W)$ be the output of the encoder. Since the code is $\epsilon$-private it satisfies $I(U;X)= \epsilon$. Due to recoverability, we can recover the useful data, it also satisfies $H(Y|U,X,W)=0$. Thus, by using Lemma~\ref{converse} we have
$
H(\tilde{U})\geq \max\{L_1(\epsilon,P_{XY}),L_2(\epsilon,P_{XY}),L_3(\epsilon,P_{XY})\},
$
which results in \eqref{23}. The proof of \eqref{ma} is similar to \cite[Theorem~9]{kostala}. For any $\tilde{u}\in \xi$ we have
\begin{align*}
P_{\tilde{U}}(\tilde{u})&=\sum_x P_X(x)P_{\tilde{U}|X}(\tilde{u}|x)\\&\leq \max_xP_X(x)\sum_x P_{\tilde{U}|X}(\tilde{u}|x)\\&=\max_xP_X(x) \sum_x\sum_w P_{\tilde{U}|WX}(\tilde{u}|w,x)P_{W}(w)
\\&=\max_xP_X(x)\sum_w P_{W}(w) \sum_x P_{\tilde{U}|WX}(\tilde{u}|w,x)\\ &\stackrel{(a)}{\leq}\max_xP_X(x)\sum_w P_{W}(w)=\max_xP_X(x),
\end{align*}
where (a) follows by the lossless condition, since as argued in \cite[Theorem~9]{kostala}, for a lossless code $P_{\tilde{U}|WX}(\tilde{u}|w,x)>0$ for at most one $x\in\cal X$. Hence, we have
$
H(\tilde{U})\geq \log(\frac{1}{\max_x P_X(x)}).
$

\end{document}